\documentclass[aps,pra,showpacs,twocolumn,nofootinbib,10pt, superscriptaddress]{revtex4-1}

\usepackage[dvips]{graphicx}
\usepackage{epstopdf}
\usepackage{amsmath}
\usepackage{amssymb}
\usepackage{mathrsfs}
\usepackage{amsthm}
\usepackage{bm}
\usepackage{url}
\usepackage[T1]{fontenc}
\usepackage{csquotes}
\MakeOuterQuote{"}

\newcommand{\sM}{\mathsf{M}}

\newtheoremstyle{note}
  {\topsep/2}               
  {\topsep/2}               
  {}                      
  {\parindent}            
  {\itshape}              
  {.}                     
  {5pt plus 1pt minus 1pt}
  {}

\theoremstyle{note}
\newtheorem{theorem}{Theorem}
\newtheorem{lemma}{Lemma}

\newtheorem{corollary}{Corollary}
\newtheorem{proposition}{Proposition}

\theoremstyle{definition}

\theoremstyle{remark}

\newenvironment{proofof}[1]{\vspace*{5mm} \par \noindent
         {\it Proof of #1:\hspace{2mm}}}{\endproof \hfill$\Box$ 
\vspace*{3mm}}

\def\cH{\mathcal{H}}

\newcommand{\mrm}[1]{\mathrm{#1}}

\newcommand{\tr}{\operatorname{tr}}

\newcommand{\diag}{\operatorname{diag}}

\newcommand{\rmc}{\mathrm{c}}

\newcommand{\rmi}{\mathrm{i}}

\newcommand{\rmr}{\mathrm{r}}

\newcommand{\rmD}{\mathrm{D}}

\newcommand{\caH}{\mathcal{H}}

\newcommand{\be}{\begin{equation}}
\newcommand{\ee}{\end{equation}}
\newcommand{\ba}{\begin{align}}
\newcommand{\ea}{\end{align}}

\def\<{\langle}  
\def\>{\rangle}  

\newcommand{\thref}[1]{Theorem~\ref{#1}}
\newcommand{\Thref}[1]{Theorem~\ref{#1}}
\newcommand{\thsref}[1]{Theorems~\ref{#1}}

\newcommand{\cref}[1]{Conjecture~\ref{#1}}
\newcommand{\Cref}[1]{Conjecture~\ref{#1}}

\newcommand{\rcite}[1]{Ref.~\cite{#1}}

\newcommand{\rmf}{\mathrm{F}}

\def\rE{\mathbb{E} }
\def\Pr{\mathbb{P} }

\def\sE{\mathsf{E}}

\def\Label#1{\label{#1}\ [\ \text{#1}\ ]\ }
\def\Label{\label}

\begin{document}

\title{Secure uniform random number extraction via incoherent strategies}

\author{Masahito Hayashi}
\affiliation{Graduate School of Mathematics, Nagoya University, 
Nagoya, 464-8602, Japan}
\affiliation{Centre for Quantum Technologies, National University of Singapore, 
3 Science Drive 2, 117542, Singapore}
\email{masahito@math.nagoya-u.ac.jp}

\author{Huangjun Zhu}
\affiliation{Institute for Theoretical Physics, University of Cologne,
Cologne 50937, Germany}
\email{hzhu1@uni-koeln.de}

\begin{abstract}
To guarantee the security of uniform random numbers generated by a quantum random number generator,
we study secure extraction of uniform random numbers when the environment of a given quantum state is controlled by the third party, the eavesdropper.
Here we restrict our operations to incoherent strategies that are composed of the measurement on the computational basis and  incoherent operations (or incoherence-preserving operations). 
We show that the maximum secure extraction rate is equal to the relative entropy of coherence.
By contrast, the coherence of formation gives the extraction rate when a certain constraint is imposed on  eavesdropper's operations. The condition under which the two extraction rates coincide is  then determined. Furthermore, we find that the exponential decreasing rate of the leaked information is characterized by R\'{e}nyi relative entropies of coherence.
These results clarify the power of incoherent strategies in random number generation, and
can be applied to guarantee the quality of random numbers generated by a quantum random number generator.
\end{abstract}

\date{\today}
\maketitle

Recently, quantum random number generation attracts much attention 
because of many practical  applications, such as cryptography, scientific simulation, and foundational studies  \cite{Ma16,HG17}.
A quantum random number generator is a device for extracting secure uniform random numbers from  quantum states. 
Its experimental demonstration has been done with quantum optics \cite{Jennewein,Furst,Ren,Wei}.
Ideally, the random numbers generated  should be independent of the third party, the eavesdropper (Eve). 
In practice, however, the relevant states or
random numbers  are often correlated to Eve.
For this reason, it is crucial to extract secure uniform random numbers from random numbers whose side information is leaked to Eve.  This task is called secure uniform random number extraction,  which has been studied in the framework of information  security,
and has been considered as a basic tool for quantum key distribution \cite{Renner,marco,RW}.
 Here the goal of  the legitimate user, Alice, is to generate random numbers that are almost independent of Eve.
Usually, the initial state is taken to be a classical-quantum (C-Q) state $\rho_{AE}$, in which, Alice's information is given as a classical random number,
while Eve's information is given as a quantum state that is correlated  to Alice's random variable. When the $n$-tensor product state $\rho_{AE}^{\otimes n}$ is given, 
it is known that the asymptotic secure extraction rate is equal to the conditional entropy
$H(A|E)_{\rho_{AE}}:=S(\rho_{AE})-S(\rho_E)$.

\begin{figure}
	\includegraphics[width=6cm]{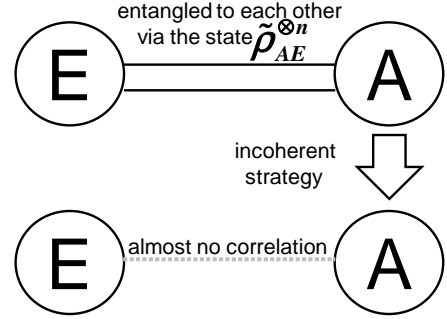}
	\caption{Extraction of secure uniform random umbers via incoherent strategy.}
\Label{FT}
\end{figure}

To guarantee the quality of the random numbers generated in a quantum random number generator, it is usually  assumed that the environment of Alice's system is controlled by Eve. This convention covers the most powerful Eve and is typical in similar research areas.
For example, in the study of quantum key distribution \cite{SP,marco,HQKD} and private capacity \cite{LAZG}, all of the environment is assumed to be under Eve's control.
Further, since Alice generates the quantum state on her system, which is under her control,
it is natural to treat Alice's initial information as a quantum state in the same way as Eve's state.
In fact, there are several  formulations of secure uniform random number extraction
with quantum-quantum states \cite[Section 4.3]{Berta}\cite{Berta2,BFW13,Faw12}.
However, little is known about the optimal extraction rate.
This is because it is not easy to clarify the reasonable range of allowed operations.

To extract  random numbers from a quantum state, Alice can  apply a projective measurement. However, not all quantum states can produce secure random numbers in this way given that the environment is controlled by Eve. Quantum coherence with respect to the measurement basis is crucial to realizing the independence from the environment and the randomness of the outcome simultaneously \cite{Ma16,HG17}.
In addition, in many practical scenarios, it is not easy to create or increase coherence in  quantum systems \cite{MarvS16,StreAP16}. Understanding the limit of random number generation in such practical scenarios is thus of paramount interest not only to theoretical study, but also to real applications. 
Although coherence is indispensable in many applications, such as laser and quantum metrology, the resource theory of  coherence was not established until recently \cite{Aber06,BCP,WintY16,StreAP16,Chitambar16PRL, HuHPZ17}. Under this framework, Yuan et al. \cite{YZCM,YZCM2} showed that the amount of randomness upon measurement on the computational basis is closely related to several important coherence measures, such as the relative entropy of coherence and coherence of formation. 
The relation between our paper and  \cite{YZCM} is explained in more detail in Appendix \ref{Relation}. However, the extraction of uniform random numbers under general  incoherent operations has not been discussed.

Motivated by the problem mentioned above,
in this paper we study the secure  extraction of uniform random numbers under  incoherent strategies, which
include the measurement on the computational basis and  general incoherent operations (or incoherence-preserving operations) \cite{Aber06,BCP,WintY16,StreAP16}. To guarantee the security of the random numbers generated, we 
assume that the environment of the relevant quantum state is controlled by  Eve; see Fig.~\ref{FT}.
We  show that the maximum secure extraction rate 
is equal to the relative entropy of coherence. By contrast, the extraction rate coincides with the coherence of formation if Eve's operations are  constrained in a special way. 
The condition under which the extraction rates in the two scenarios coincide has a simple description. Furthermore, we show that the exponential decreasing rate of the  leaked information is characterized by R\'{e}nyi relative entropies of coherence.
These results not only clarify the power of incoherent strategies in extracting random numbers, but also endow operational meanings to a number of important coherence measures.

The resource theory of coherence is characterized by the set of incoherent states, denoted by $\mathcal{I}$, and the set of incoherent operations \cite{Aber06,BCP,WintY16,StreAP16}.
Recall that a state is incoherent if it is diagonal with respect to the reference computational basis. A quantum operation, represented by a completely positive trace preserving (CPTP) map, is incoherence-preserving (also called maximally incoherent) if it maps  incoherent states to incoherent states \cite{Aber06}. It is incoherent if, in addition, each Kraus operator in  its Kraus representation maps incoherent states to incoherent states up to normalization \cite{BCP}. An incoherent operation is physically incoherent if it admits an incoherent Stinespring dilation \cite{Chitambar16PRL}. For unitary transformations, the three types of operations coincide.  
A unitary operator is  incoherent if and only if  (iff) each row and each column has only one nonzero entry.

The relative entropy of coherence $C_{\rmr}(\rho) $ of a quantum state $\rho$ is the minimum relative entropy between the state and any incoherent state \cite{BCP,Aber06},
\begin{align}\label{eq:REC}
C_{\rmr}(\rho):=
\min_{\sigma \in {\cal I}} S(\rho\|\sigma)=S(\rho^{\diag} ) -S(\rho),
\end{align}
where $S(\rho\|\sigma):= \tr \rho (\log\rho -\log \sigma)$ is the relative entropy between $\rho$ and $\sigma$, $S(\rho)$ is the von Neumann entropy of $\rho$, and $\rho^{\diag}$ is the  diagonal part of $\rho$. In this paper "log" has base 2.
The coherence of formation $C_{\rmf}(\rho) $ is the convex roof of $C_{\rmr}(\rho)$ \cite{Aber06,YZCM},
 \begin{align}
C_{\rmf}(\rho)
:=\min_{ \{p_j, |\psi_j\rangle \}}
\sum_j p_j 
C_{\rmr}( |\psi_j\rangle \langle |\psi_j|),
\end{align}
where $\{p_j, |\psi_j\rangle\}$ satisfies  $ \rho=\sum_j p_j  |\psi_j\rangle \langle \psi_j|$. 
It is known that the relative entropy of coherence $C_{\rmr}(\rho)$ is equal to the distillable coherence, and the coherence of formation $C_{\rmf}(\rho) $ is equal to the coherence cost \cite{WintY16}. 


In practice, Alice repeatedly  generates many copies of identical and independent quantum states.
This assumption allows us to write the state of the whole
system as a tensor product, so  our problem can be formulated as follows.
Suppose Alice holds $n$ copies of the quantum state $\tilde{\rho}_A$ on
system $\mathcal{H}_A$  whose environment is controlled by Eve. 
All the information of Eve about Alice's systems is encoded in a purification, say 
$\tilde{\rho}^{\otimes n}$, of $\tilde{\rho}_A^{\otimes n}$.
Alice  is  allowed  to perform only incoherent strategies, which can be divided into three steps without loss of generality.
First, she applies an incoherent unitary operation  $U_{\rmi,n}$ on the system and an ancilla system $\cH_B$, whose initial state is $|0\rangle$.
Second, she performs the measurement $\sM_{\rmc,n}$ on the computational basis, whose set of outcomes is denoted by ${\cal A}^n$.
Finally, as  post-measurement processing, she applies a random hash function $F_n$ from 
${\cal A}^n $ to a suitable set ${\cal L}_n$.
The cardinality (number of elements) of  ${\cal L}_n$ is denoted by $|{\cal L}_n|$, which also expresses the dimension of the output system.
In this paper, a random variable is denoted by an italic capital letter, and 
its probability space by the  same letter in mathcal font. 
The incoherent strategy of Alice is characterized by the triple $(U_{\rmi,n}, \sM_{\rmc,n}, F_n )$ 
and is denoted by $\sM_{F_n}$ for simplicity.
The cardinality $|{\cal L}_n|$ is also denoted by $|\sM_{F_n}|$.

To determine the maximum extraction rate of secure uniform random numbers, we need a security measure.
When the whole system  is characterized by a  C-Q state ${\rho}_{AE}$,
a widely accepted  measure on secure random numbers is 
the trace norm (also known as the Schatten  1-norm) between the real state and the ideal state,
\begin{align}
d_1({\rho}_{AE}):=\|
{\rho}_{AE}-\tau_{|{\cal A}|} \otimes {\rho}_{E}
\|_1,
\end{align}
where $\tau_{V}$ is the completely mixed state on the $V$-dimensional system.
So, $\tau_{|{\cal A}|}$ expresses the completely mixed state on $\cH_A$. The significance of this measure lies in  the fact that
it is universally composable  \cite{Renner,RW}.

Here the security measure of concern is the value
$d_1(\sM_{F_n}|F_n):=
\rE_{F_n} d_1(\sM_{F_n} (\tilde{\rho}^{\otimes n})) $,
where $\rE_{F_n}$ expresses the expectation with respect to the choice of the random hash function.
The maximum asymptotic extraction rate of secure uniform random numbers 
$R(\tilde{\rho}_A)$ from the  $n$-tensor product 
$\tilde{\rho}_A^{\otimes n}$ is defined as
\begin{align}
R(\tilde{\rho}_A):=
\max_{
\{\sM_{F_n}\}
}
\Big\{ 
\liminf_{n\to \infty}
\frac{\log |\sM_{F_n}|}{n} \Big|
d_1(\sM_{F_n}|F_n) \to 0
\Big\},
\end{align}
where the maximum is taken over  sequences of incoherent strategies $\sM_{F_n}$ which satisfy the given condition.

To compute the rate $R(\tilde{\rho}_A)$, 
we need to study the uncertainty of Alice's system from Eve's viewpoint
when the initial state on $\cH_A\otimes \cH_E$ is a pure state.
This uncertainty can be measured by the conditional entropy 
$H(A|E)_{\tilde{\rho}}=S(\tilde{\rho})-S(\tilde{\rho}_E)$.
To maximize Eve's uncertainty, Alice can introduce an ancilla system $\mathcal{H}_B$ prepared in the incoherent state $|0\rangle \langle 0|$, so that the initial state is $ \tilde{\rho}\otimes |0\rangle \langle 0|$. 
Then  she applies an incoherent unitary $U_\rmi$ on $\cH_A \otimes \cH_B$, which leads to
 the output state $\tilde{\rho}[U_\rmi]:= U_\rmi (\tilde{\rho}\otimes |0\rangle \langle 0|) U_\rmi^\dagger $. When $d_B:=\dim(\caH_B)\geq d_A:=\dim(\caH_A)$, a particularly interesting incoherent unitary is the  the generalized CNOT gate defined as 
 \begin{equation}
 U_{\mrm{CNOT}}:= \sum_{x;\; y<d_A}| x,x+y \rangle \langle x,y| +\sum_{x;\; y\geq d_A}| x,y \rangle \langle x,y|,
 \end{equation}
 where the addition $x+y$ is modulo $d_A$. 
\begin{theorem}\Label{TH1}
	\begin{align}
	\frac{1}{n}\max_{U_\rmi} H(A|E)_{ \tilde{\rho}^{\otimes n}[U_\rmi]}
=&\max_{U_\rmi} 
H(A|E)_{ \tilde{\rho}[U_\rmi]}\nonumber \\
=& 
H(A|E)_{ \tilde{\rho}[U_{\mrm{CNOT}}]}
	= C_{\rmr}(\tilde{\rho}_A)\label{eq:H266},
	\end{align}
	where $U_\rmi$ is an incoherent unitary.
\end{theorem}
\begin{proof}
Let $U_\rmi$ be any incoherent unitary. Then 	
\begin{align}
	 H(A|E)_{\tilde{\rho}[U_\rmi]} &=-H(A|B)_{\tilde{\rho}[U_\rmi]} 
	\le  E_\rmr(\tilde{\rho}[U_\rmi]_{AB})
	\nonumber\\
	&\leq {C}_\rmr(\tilde{\rho}[U_\rmi]_{AB})\leq {C}_\rmr(\tilde{\rho}_A).
	\Label{eq:H271}
	\end{align}
Here the equality follows from the duality relation $H(A|E)_{\rho}+H(A|B)_{\rho}=0$, which holds whenever $\rho$ is pure; the first inequality follows from \cite{PlenVP00}\cite[Lemma~4]{ZHC},  the second inequality from the fact that incoherent states for a bipartite system are separable, and the third inequality from the fact that the relative entropy of coherence is monotonic under  incoherence-preserving operations. 

According to the duality relation and  \cite{StreSDB15}\cite[Theorem~1]{ZHC}, 
$H(A|E)_{\tilde{\rho}[U_{\mrm{CNOT}}]}
=-H(A|B)_{\tilde{\rho}[U_{\mrm{CNOT}}]} ={C}_\rmr(\tilde{\rho}_A)$, note   that $\tilde{\rho}[U_{\mrm{CNOT}}]_{AB}$ is a maximally correlated state \cite{Rain99,ZhuMCF17,ZhuHC17}.  Therefore, $\max_{U_\rmi} H(A|E)_{ \tilde{\rho}[U_\rmi]} \nonumber 
	= C_{\rmr}(\tilde{\rho}_A)$. Now the proof of \eqref{eq:H266} is completed by the additivity relation  $C_{\rmr}(\tilde{\rho}_A^{\otimes n})= nC_{\rmr}(\tilde{\rho}_A)$ \cite{WintY16,ZHC}.
\end{proof}

\Thref{TH1} is helpful for computing the extraction rate
$R(\tilde{\rho}_A)$ as follows.
If Alice performs the measurement $\sM_{\rmc}$ in the computational basis,
then $\tilde{\rho}$ is turned into the state
 $\sM_{\rmc}(\tilde{\rho}):=
\sum_{x} |x\rangle \langle x |\otimes
\langle x| \tilde{\rho}|x\rangle$, which satisfies
\begin{align}\label{eq:McCNOT}
H(A|E)_{\sM_\rmc(\tilde{\rho})}=H(A|E)_{ \tilde{\rho}[U_{\mrm{CNOT}}]}=C_{\rmr}(\tilde{\rho}_A).
\end{align}
After repeating this procedure and generating the state 
$\sM_\rmc(\tilde{\rho})^{\otimes n}$, Alice applies a random hash function $F_n$ to the $n$ measurement outcomes with 
the  extraction rate of uniform random numbers chosen to be $R$. Here the random hash function $F_n$ is assumed to satisfy the {\it universal 2 condition} as discussed  in Appendix \ref{A2A}
 \cite{Carter,WC81}, 
which is  conventional in generating secure random numbers from 
random numbers that might be partially leaked to the eavesdropper.
The  efficient construction of such hash functions was discussed in \cite{HT2}.
In the independent and identical situation, 
Proposition~\ref{LH3} in Appendix~\ref{A2A} shows that the extracted random numbers are secure  when the extraction rate $R$ is smaller than the conditional entropy 
$H(A|E)_{\sM_\rmc(\tilde{\rho})}$. 
Therefore, we have
$R(\tilde{\rho}_A) \ge 
H(A|E)_{\sM_\rmc(\tilde{\rho})}$. 
Since Alice can optimize the incoherent unitary
before the measurement $\sM_\rmc$,
it follows  that 
\begin{align}
R(\tilde{\rho}_A) \ge  
\max_{U_{\rmi}} H(A|E)_{\tilde{\rho}[U_\rmi]} \Label{r2}.
\end{align}

Conversely, as shown in Appendix \ref{AB1}, the opposite inequality
\begin{align}
R(\tilde{\rho}_A) \le 
\liminf_{n \to \infty}
\frac{1}{n}
\max_{U_{\rmi}}
H(A|E)_{\tilde{\rho}^{\otimes n}[U_{\rmi}] } 
\Label{r1}
\end{align}
holds.
Combining \eqref{eq:H266} of Theorem \ref{TH1} with \eqref{r2} and \eqref{r1},
we obtain the following theorem.
\begin{theorem}\Label{TH3}
The extraction rate $R(\tilde{\rho}_A)$ is given by
\begin{align}
R(\tilde{\rho}_A)=
\max_{U_{\rmi}} H(A|E)_{ \tilde{\rho}[U_\rmi]} 
= C_{\rmr}(\tilde{\rho}_A)
\Label{eq:H26T},
\end{align}
where $U_{\rmi}$ is an incoherent unitary.
\end{theorem}
According to \eqref{eq:McCNOT}, \eqref{r2}, and \eqref{r1}, the maximum extraction rate $C_{\rmr}(\tilde{\rho}_A)$ stated in \thref{TH3} can be achieved by the measurement $\sM_\rmc$ on the computational basis (without other incoherent operations) followed by classical data processing characterized by $F_n$. This strategy is denoted by $\sM_{F_n}^*$ henceforth.

\begin{figure}
	\includegraphics[width=6cm]{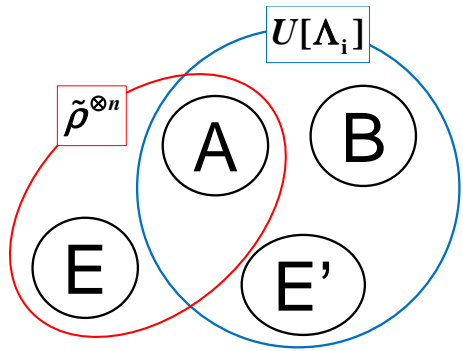}
	\caption{Extended strategy.
Alice can apply a general incoherent (or incoherence-preserving) operation $\Lambda_{\rmi}$.
Both $\cH_E$ and $\cH_{E'}$ are in Eve's hands.}
\Label{F2}
\end{figure}

Now, we extend Alice's incoherent unitaries to general  incoherence-preserving CPTP maps acting on the system $\cH_A \otimes \cH_B$.
If she uses a CPTP map whose final state is always the specific incoherent state 
$\sum_{i=0}^{d-1}\frac{1}{d}|i \rangle \langle i| $,
the resulting conditional entropy equals $\log d$, which increases unlimitedly as $d$ increases.
To avoid such a trivial advantage for Alice,
similar to the study of 
quantum key distribution \cite{SP,marco,HQKD} and private capacity \cite{LAZG},
we assume that the environment $\cH_{E'}$ of the incoherence-preserving CPTP map $\Lambda_{\rmi}$ is also controlled by Eve, so that
 Eve has the two systems $\cH_E$ and $\cH_{E'}$ in total.
This is because it is not easy to exclude the possibility that Eve accesses a  system that interacts with  Alice's operation.
To cover such a worst scenario, we take this convention and consider the Stinespring representation 
$\rho_{E'}[\Lambda_{\rmi}],U[\Lambda_{\rmi}]$ of $\Lambda_{\rmi}$, where 
$\rho_{E'}[\Lambda_{\rmi}]$ is the initial  pure state on the environment
and $U[\Lambda_{\rmi}]$ is the  unitary on the whole system. Note that $U[\Lambda_{\rmi}]$ may not be incoherent if $\Lambda_\rmi$ is not physically incoherent \cite{Chitambar16PRL}, but this fact does not affect the following argument.
Now the total output state is 
$\tilde{\rho}[\Lambda_{\rmi}]
:=U[\Lambda_{\rmi}] (\tilde{\rho} \otimes |0\rangle \langle 0|\otimes \rho_{E'}[\Lambda_{\rmi}]
) U[\Lambda_{\rmi}]^\dagger $.
Since $\tilde{\rho}[\Lambda_{\rmi}]$ is a pure state, 
we can still use the duality relation on conditional entropies.
So, similar to \eqref{eq:H271}, we have
\begin{align}
 H(A|E)_{\tilde{\rho}[\Lambda_{\rmi}]}
\le E_\rmr( \tilde{\rho}[\Lambda_{\rmi}]_{AB})\leq C_{\rmr}(\tilde{\rho}_A).\Label{eq:H263B}
\end{align}
Again, the two inequalities are saturated when $\Lambda_{\rmi}$ is the generalized CNOT gate. 
Therefore, \thsref{TH1} and \ref {TH3} still hold if incoherent unitaries are replaced by general  incoherence-preserving operations.



Here, we need to discuss the relation with the distillable coherence $C_\rmD(\tilde{\rho}_A)$ under incoherent operations, which is equal to $C_\rmr(\tilde{\rho}_A)$ \cite{WintY16}.
Note that coherence distillation may require incoherent operations  across  many copies, and these operations may not be
physically incoherent. 
By contrast, to implement our optimal protocol, it suffices to perform the measurement $\sM_{\rmc}$ 
followed by classical data processing, i.e., application of  universal 2 hash functions, which is much easier.

Now, we remember that the criterion $d_1$ universally covers the distinguishability by Eve's local measurement $\sM_E$.
Since the criterion $d_1$ is universally composable, 
the above discussion covers the case in which Eve chooses her local measurement $\sM_E$
according to the choice of the hash function $f_n$.
Now, we consider the scenario in which 
Eve cannot choose her local measurement $\sM_E$
according to the random choice of the hash function $F_n$, although she knows which hash function $F_n$ is applied after her measurement $\sM_E$.
Here $f_n$ denotes a specific hash function, while $F_n$ denotes a random hash function.
Given the $n$ tensor product state $\tilde{\rho}^{\otimes n}$, 
we introduce a new security criterion 
$\underline{d}_1( \sM_{F_n}|F_n)$ as
\begin{align}
&\underline{d}_1(\sM_{F_n}|F_n)
:=
\max_{\sM_E} \rE_{F_n} d_1\bigl(\sM_{F_n}(\sM_E (\tilde{\rho})^{\otimes n})\bigr) \nonumber \\
=&
\max_{\sM_E} \rE_{F_n} d_1\bigl(\sM_{F_n}(\sM_E^{\otimes n} (\tilde{\rho}^{\otimes n}))\bigr) 
\le 
 d_1(\sM_{F_n}|F_n)
\Label{e1-1},
\end{align}
where $\sM_E$ is Eve's POVM on the  system $\cH_E$.
Then, instead of $R(\tilde{\rho}_A)$, we define
\begin{align}
&\overline{R}(\tilde{\rho}_A):=
\max_{\{\sM_{F_n}\}}
 \left\{ \left.
\liminf_{n\to \infty}
\frac{\log |\sM_{F_n}|}{n} \right|
\underline{d}_1(\sM_{F_n}|F_n)
\to 0
\right\},\Label{29-11}
\end{align}
where the maximum is taken over
 sequences of  incoherent strategies $\sM_{F_n}$ which satisfy the given condition.
The relation \eqref{e1-1} implies the inequality $\overline{R}(\tilde{\rho}_A)\ge {R}(\tilde{\rho}_A)$.
Instead of Theorem \ref{TH3}, we have the following theorem:
\begin{theorem}\Label{TH7}
\begin{align}
\overline{R}(\tilde{\rho}_A)=
C_{\rmf}(\tilde{\rho}_A)
\Label{eq:H26T3}.
\end{align}
\end{theorem}
This theorem offers an operational meaning of the coherence of formation $C_{\rmf}(\tilde{\rho}_A)$.
Since the relation $C_{\rmf}(\tilde{\rho}_A) \ge C_{\rmr}(\tilde{\rho}_A) $ holds in general and the inequality is generically strict, Theorems \ref{TH3} and \ref{TH7} show that Alice can usually extract secure uniform random numbers with a higher  rate
if Eve chooses her measurement independently of
the incoherent strategies of Alice.
In conjunction with Theorem~10 in \cite{WintY16}, we can  deduce the condition under which the rates in the two scenarios coincide.
\begin{theorem}\label{thm:twoRates}The inequality 
$\overline{R}(\tilde{\rho}_A)\geq R(\tilde{\rho}_A)$ is saturated iff $\tilde{\rho}_A$ is pure or its eigenvectors are supported on  orthogonal subspaces
spanned by a partition of basis states in the reference basis. 
\end{theorem}
The following corollary is an easy consequence of \thref{thm:twoRates}; a direct proof is presented in the appendix.
\begin{corollary}\Label{cor:HZL}
 	A qubit state $\tilde{\rho}_A$ saturates the inequality
 	$\overline{R}(\tilde{\rho}_A)\geq R(\tilde{\rho}_A)$ iff $\tilde{\rho}_A$ is pure or incoherent.
 \end{corollary}

In many topics of quantum information,
the R\'{e}nyi entropies characterize the exponential decreasing rate of the error probability, which  determines the speed of  convergence \cite{c-expo,s-expo}. 
Concerning secure uniform random number generation,
it is known that the exponential decreasing rate of the leaked information
is characterized  by R\'{e}nyi conditional entropies, as explained in Appendix~\ref{A2A}.
To determine the speed of  convergence 
$d_1(\sM_{F_n} |F_n) \to 0$,
we introduce the R\'{e}nyi relative entropy of coherence
$\underline{C}_{\rmr,\alpha}(\rho):=\min_{\sigma\in \mathcal{I}}\underline{S}_\alpha(\rho\|\sigma)$ \cite{Chitambar,ZHC}
based on the R\'{e}nyi relative entropy
$\underline{S}_\alpha(\rho\|\sigma):= 
\frac{1}{\alpha-1}\log \tr\bigl( \sigma^{\frac{1-\alpha}{2\alpha}}\rho\sigma^{\frac{1-\alpha}{2\alpha}} \bigr)^{\alpha}$ with $\alpha\geq0$ \cite{MullDSF14,WildWY14}\cite{Ha3}\cite[Theorem 5.13]{Haya17book}.
Combining Proposition \ref{LH1} in Appendix \ref{A2A} with a generalization of Theorem~\ref{TH1} in Appendix \ref{AK}, 
we can derive the following theorem, whose proof is relegated to Appendix \ref{ATH3}.
\begin{theorem}\Label{TH3X}
Suppose that $F_n$ are universal 2 hash and have extraction rate $R$.
Then 
	\begin{align}
	\liminf_{n \to \infty}\frac{-1}{n}\log 
	d_1(\sM_{F_n}^* |F_n) 
	\ge  \max_{s \in [0,1]}\frac{s}{2} \bigl(
	\underline{C}_{\rmr,\frac{1+s}{1+2s}}(\tilde{\rho}_A)-R\bigr).
	\Label{Eq4-20B}
	\end{align}
\end{theorem}

\Thref{TH3X} shows that the exponential decreasing rate of the leaked information of the strategy $\sM_{F_n}^*$ is characterized  by R\'{e}nyi relative entropies of coherence.
In other words, the quality of the random numbers extracted in this way is controlled by these coherence measures.
In information theory, another useful security measure is the relative entropy between the true state and the ideal state \cite{CN}, which is known to be the unique measure under several natural assumptions \cite[Theorem 8]{epsilon}. In 
Appendix~\ref{AAT} we show that an analog of \thref{TH3X} holds for this alternative measure.

{\it Conclusion:\quad}
We  studied the extraction of secure uniform random numbers in the quantum-quantum setting
under incoherent strategies, assuming that
Eve can access all of the environment of the given system.
This problem properly reflects the situation of a quantum random number generator.
We  showed that
the maximum rate of extraction is equal to the relative entropy of coherence.
In contrast, the extraction rate with a constrained eavesdropper  
is equal to the coherence of formation.
Furthermore, the exponential decreasing rate of the leaked information is characterized by R\'{e}nyi relative entropies of coherence.
These results not only clarify the capability
of incoherent strategies in extracting secure uniform random numbers, but also endow coherence measures mentioned above with operational meanings.

To apply our results to the security evaluation of a quantum random number generator, 
we need to estimate the quantum state $\tilde{\rho}_A$ on Alice's system priorly.
Fortunately, as explained in Appendix~\ref{industrial}, this task can be achieved by
 quantum state tomography, which has been well established 
\cite{Hol,Gill,HayaB}\cite[Chapter 6]{Haya17book}.
Even when  Alice's quantum system cannot be 
trusted,
we can estimate the quantum state $\tilde{\rho}_A$ of Alice's system 
by combining the method of self testing \cite{MY2,McKague1,HH}.
Therefore,
our study is helpful to the design of 
a quantum random number generator; see Appendix~\ref{industrial} for more details.

\acknowledgments
MH was supported in part by JSPS Grants-in-Aid for Scientific Research (A) No.17H01280 and (B) No. 16KT0017 and Kayamori Foundation of Informational Science Advancement.
HZ acknowledges financial support
from the Excellence
Initiative of the German Federal and State Governments
(ZUK~81) and the DFG.
The authors are very grateful to Professor Lin Chen for helpful discussions and
comments.

\appendix

\section{Secure uniform random number extraction from a C-Q state}\Label{A2A}
Here, we summarize  known results on
secure uniform random number extraction
when the state is a C-Q state on the composite system $\cH_A\otimes \cH_E$, which has the form
\begin{equation}
\rho_{AE}=\sum_{a} P_A(a) |a\rangle \langle a| \otimes \rho_{E|a}.
\end{equation}
Given a function $f$, 
we define the state
\begin{equation}
\rho_{f(A)E}:=
\sum_{a} P_A(a) |f(a)\rangle \langle f(a)| \otimes \rho_{E|a}.
\end{equation}
 To study  secure uniform random number extraction from a C-Q state,
we need to  consider the uncertainty quantified by three types of  R\'enyi  conditional entropies,
\begin{align}
\overline{H}_\alpha^{\uparrow}(A|E)_{\rho}
&:=-\min_{\sigma_E} \underline{S}_\alpha( \rho_{AE} \| I_A \otimes \sigma_E) \Label{U1}, \\
\overline{H}_\alpha^{\downarrow}(A|E)_{\rho}
&:=- \underline{S}_\alpha( \rho_{AE} \| I_A \otimes \rho_{E}), \\
{H}_\alpha^{\uparrow}(A|E)_{\rho}
&:=-\min_{\sigma_E} {S}_\alpha( \rho_{AE} \| I_A \otimes \sigma_{E}).
\end{align}
Here the two types of R\'{e}nyi relative entropies are defined as
\cite{MullDSF14,WildWY14}
\cite[Section 3.1]{Haya17book}
\begin{align}
S_\alpha(\rho\|\sigma)&:= \frac{1}{\alpha-1}\log \tr (\rho^{\alpha}\sigma^{1-\alpha}), \label{eq:RREa}\\
\underline{S}_\alpha(\rho\|\sigma)&:= 
\frac{1}{\alpha-1}\log \tr\bigl( \sigma^{\frac{1-\alpha}{2\alpha}}\rho\sigma^{\frac{1-\alpha}{2\alpha}} \bigr)^{\alpha}, \label{eq:RREb}
\end{align}
which satisfy the inequality $S_\alpha(\rho\|\sigma) \ge
\underline{S}_\alpha(\rho\|\sigma)$.  Both $S_\alpha(\rho\|\sigma)$ and $\underline{S}_\alpha(\rho\|\sigma)$ increase monotonically with $\alpha$.

To extract secure uniform random numbers, we can employ a universal 2 hash function.
A random function $F$ from ${\cal A}$ to ${\cal Z}$ 
is called universal 2 hash if
\begin{align}
\Pr \{ F(a) =F(a')\} \le \frac{1}{|{\cal Z}|}
\end{align}
for $a\neq a' \in {\cal A}$. 
This type of hash functions satisfy the following leftover hashing lemma.
\begin{proposition}[\protect{\cite{Renner}}]\Label{LR}
Let $F$ be a universal 2 hash function from ${\cal A}$ to ${\cal Z}$.
Then, we have
\begin{align}
\rE_F d_1(\rho_{F(A)E} )
\le 
|{\cal Z}|^{\frac{1}{2}} 
2^{-\frac{1}{2}\overline{H}^{\uparrow}_2(A|E)_{\rho_{AE}}}.
\end{align}
\end{proposition}

To characterize the ultimate amount of extracted secure uniform random numbers,
we define the rate
\begin{align}
&K(\rho_{AE}):=\nonumber \\
&\sup_{F_{n}} \Big\{ 
\liminf_{n\to \infty}
\frac{\log |F_n|}{n} \Big|
\rE_{F_n} d_1\bigl( (\rho^{\otimes n})_{F_{n}(A)E}\bigr)\to 0
\Big\},
\end{align}
where $|F_{n}| $ denotes  the cardinality of the image of  $F_{n}$ and the supremum is taken over sequences of  random hash functions which satisfy the given condition.
The quantity $K(\rho_{AE})$ expresses the maximum extraction rate of secure uniform random numbers.

In this setting, the simple application of Proposition~\ref{LR} cannot guarantee the exponential decrease of the leaked information
even when the extraction rate $R$ of uniform random numbers is smaller than the conditional entropy $H(A|E)_{\rho_{AE}}$.
To resolve this problem, we employ another proposition based on the
discussions in \cite{Ha1}, which in turn rely on Proposition~\ref{LR}.

\begin{proposition}\Label{LH1}
If a sequence of hash functions
$F_{n}$ from ${\cal A}^{n}$ to $\{1, \ldots, 2^{nR}\}$
is universal 2 hash, then
\begin{align}
&\liminf_{n \to \infty}
-\frac{1}{n}\log 
\rE_{F_n} d_1\bigl( (\rho^{\otimes n})_{F_{n}(A)E} \bigr)\nonumber \\
\ge & \max_{s \in [0,1]}
\frac{1}{2} \bigl(s \overline{H}_{1+s}^{\uparrow}( A|E)_{\rho_{AE}} -sR\bigr).
\Label{Eq4-20AC}
\end{align}
\end{proposition}

\begin{proof}
First, we introduce the quantity \cite[Section IV]{Ha1}
\begin{align*}
&\Delta_{d,2}( M |\rho_{AE}) \nonumber \\
&:=\min_{\sigma_E}
\min_{\rho_{AE}'}
\Bigl[
2 \|\rho_{AE}-\rho_{AE}'\|_1 
+ M^{\frac{1}{2}}2^{\frac{1}{2}\underline{S}_{2}(\rho_{AE}'\|I_A \otimes \sigma_E)} \Bigr],
\end{align*}
where $\min_{\rho_{AE}'}$ denotes the minimum under the condition
$\tr \rho_{AE}'\le 1$ and $\rho_{AE}' \ge 0$, while $\min_{\sigma_{E}}$ denotes the minimum over normalized states $\sigma_{E}$.
Here, 
the quantity $\underline{S}_{2}(\rho'\|\sigma)$
is defined in the same way as in \eqref{eq:RREb}, that is,  
$\underline{S}_{2}(\rho'\|\sigma)=\frac{1}{\alpha-1}\log \tr\bigl( \sigma^{\frac{1-\alpha}{2\alpha}}\rho'\sigma^{\frac{1-\alpha}{2\alpha}}\bigr)^{\alpha}$, even when $\rho'$ is not normalized.
Then, as shown in \cite[(73)]{Ha1}, Proposition~\ref{LR} implies that
\begin{align}
\rE_{F_n} d_1\bigl( (\rho^{\otimes n})_{F_{n}(A)E} \bigr)
\le 
\Delta_{d,2}(2^{nR}|\rho_{AE}^{\otimes n}).\Label{HE1}
\end{align}
Let   $v_n$ be the number of distinct eigenvalues of $\sigma_E^{\otimes n}$. Then 
the inequality \cite[the next inequality of (83)]{Ha1} yields that
\begin{align}
&
\Delta_{d,2}(2^{nR}|\rho_{AE}^{\otimes n})
\nonumber \\
& \le (4+\sqrt{v_n}) 2^{ \frac{s}{2} n R
+\frac{s}{2} S_{1+s}\bigl(
\sE_{\sigma_E^{\otimes n}}
(\rho_{AE}^{\otimes n})
\|
I \otimes \sigma_E^{\otimes n}\bigr)
}\Label{HE2}
\end{align}
for $s \in [0,1]$,
where the CPTP map
${\sE}_{\sigma}$ is defined as
\begin{align}
{\sE}_{\sigma}(\rho):=
\sum_{x} E_x \rho E_x,
\end{align}
assuming that $\sigma$ has the spectral decomposition
$\sigma= \sum_x \lambda_x E_x$.

Since $v_n$ is a polynomial in $n$, we have 
\begin{equation}\label{eq:vnlim}
\lim_{n\to \infty}\frac{1}{n}\log v_n=0. 
\end{equation}
In addition,  
\begin{align}\label{eq:RenyiPinch}
\lim_{n \to \infty}\frac{1}{n}
 S_{1+s}\bigl(
{\sE}_{\sigma_E^{\otimes n}}
(\rho_{AE}^{\otimes n})
\|
I \otimes \sigma_E^{\otimes n}\bigr)
=
\underline{S}_{1+s}( 
\rho_{AE}\|
I \otimes \sigma_E)
\end{align}
according to \cite{MO}\cite[(3.17)]{Haya17book}. 
Combining the four equations \eqref{HE1}, \eqref{HE2}, \eqref{eq:vnlim}, and \eqref{eq:RenyiPinch} yields
\begin{align}
&\liminf_{n \to \infty}\frac{-1}{n}\log 
\rE_{F_n} d_1\bigl( (\rho^{\otimes n})_{F_n(A)E} \bigr)
 \nonumber \\
\ge & \max_{s \in [0,1]}\frac{1}{2} \bigl(
- s \underline{S}_{1+s}( 
\rho_{AE}
\|
I \otimes \sigma_E)
- s R\bigr).
\Label{Eq4-20A}
\end{align}
Taking the maximum of the right hand side (RHS) of \eqref{Eq4-20A}
over $\sigma_E$,
we obtain \eqref{Eq4-20AC}.
\end{proof}

When $R < H(A|E)$,
according to Proposition \ref{LH1},
 the amount of leaked information $\rE_{F_{n}} d_1\bigl( (\rho^{\otimes n})_{F_{n}(A)E} \bigr)$
 goes to zero.
 Hence,  we have
 \begin{align}
 K(\rho_{AE}) \ge H(A|E)_{\rho_{AE}}.
 \end{align}
Since the opposite inequality also holds \cite{Renner}\cite[(93)]{Ha1}, we deduce the following proposition.
\begin{proposition}\cite{Renner}\cite[(94)]{Ha1}\Label{LH3}
 \begin{align}
 K(\rho_{AE}) = H(A|E)_{\rho_{AE}}.
 \end{align}
\end{proposition}

In the current context, 
we often consider another security criterion $I'(\rho_{AE})$ defined as the relative entropy between the true state and the ideal state \cite[(29)]{Ha1}\cite[(9)]{Ha2},
\begin{align}\label{eq:AltSC}
I'(\rho_{AE}):= S(\rho_{AE}\|\tau_{|{\cal A}|}\otimes \rho_{E})=\log |{\cal A}|-H(A|E)_{{\rho}_{AE}},
\end{align}
where $\tau_{|{\cal A}|}$ denotes the completely mixed state on  $\caH_A$. 
Under this security criterion, we have the following analog of  Proposition~\ref{LH1}.
\begin{proposition}\cite[(33)]{Ha2}\Label{LH2}
If a sequence of hash functions
$F_{n}$ from ${\cal A}^{n}$ to $\{1, \ldots, 2^{nR}\}$
is universal 2 hash, then
\begin{align}
&\liminf_{n \to \infty}
-\frac{1}{n}\log 
\rE_{F_n} I'\bigl( (\rho^{\otimes n})_{F_{n}(A)E} \bigr)\nonumber \\
\ge & \max_{s \in [0,1]}
 \bigl(s \overline{H}_{1+s}^{\downarrow}( A|E)_{\rho_{AE}} -sR\bigr).
\Label{Eq4-20AD}
\end{align}
\end{proposition}

 \section{Proofs of \eqref{r1} and Theorem \ref{TH7}}\Label{AB1}
 \begin{proofof}{\eqref{r1}}
Let 
$\sM_{F_n}=
(U_{\rmi,n},\sM_{\rmc,n},F_n)$ be a sequence of incoherent strategies that satisfy 
$d_1(\sM_{F_n} |F_n ) \to 0$, that is, 
\begin{equation}
\rE_{F_n}\bigl\|\sM_{F_n}(\tilde{\rho}^{\otimes n})- \tau_{|\sM_{F_n}|}\otimes \sM_{F_n}
(\tilde{\rho}^{\otimes n})_E\bigr\|_1 \to 0.
\end{equation}
Then 
\begin{align}
\frac{1}{n}
\bigl|H(A|E)_{\sM_{F_n}(\tilde{\rho}^{\otimes n})}
-H(A|E)_{ \tau_{|\sM_{F_n}|}\otimes \sM_{F_n}(\tilde{\rho}^{\otimes n})_E} \bigr|
\to 0
\end{align}
according to Fannes inequality for the conditional entropy \cite[Exercise 5.38]{Haya17book}\cite{Alicki}. 
Since
$\tau_{|\sM_{F_n}|}$ is the completely mixed state on the $|\sM_{F_n}|$-dimensional system,
we have $ H(A|E)_{ \tau_{|\sM_{F_n}|}\otimes \sM_{F_n}(\tilde{\rho}^{\otimes n})_E}= \log |\sM_{F_n}|$,
which implies that
\begin{align}
\liminf_{n\to \infty}
\rE_{F_n}\frac{1}{n} H(A|E)_{\sM_{F_n}(\tilde{\rho}^{\otimes n})} =
\liminf_{n\to \infty}
\frac{1}{n}\log |\sM_{F_n}|.
\end{align}
Now \eqref{r1} is a consequence of  the following equation
\begin{align}
&\rE_{F_n} H(A|E)_{\sM_{F_n}(\tilde{\rho}^{\otimes n})} 
\le 
H(A|E)_{\sM_{\rmc,n}
(U_{\rmi,n}(\tilde{\rho}^{\otimes n}\otimes  |0\rangle \langle 0|)U_{\rmi,n}^\dagger)} \nonumber \\
&=H(A|E)_{U_{\mrm{CNOT}} U_{\rmi,n}(\tilde{\rho}^{\otimes n}\otimes  |0\rangle \langle 0|)U_{\rmi,n}^\dagger U_{\mrm{CNOT}}^\dagger} \nonumber \\
&\le
\max_{U_{\rmi}}
H(A|E)_{\tilde{\rho}^{\otimes n}[U_{\rmi}] } \Label{5-2-1}.
\end{align}
\end{proofof}

\begin{proofof}{Theorem \ref{TH7}}
If $R < \min_{\sM_E} H(A|E)_{\sM_\rmc(\sM_E(\tilde{\rho}))}$,
then Alice can extract uniform random numbers using the method described  in Appendix \ref{A2A}, and 
Proposition \ref{LH1} there guarantees that 
the extracted random numbers are secure. 
Therefore,
\begin{align}
\overline{R}(\tilde{\rho}_A) \ge &
\min_{\sM_E}H(A|E)_{\sM_\rmc(\sM_E(\tilde{\rho}))} \nonumber \\
\stackrel{(a)}{=}&
\min_{\{p_j,|\psi_j\rangle\}}
\sum_j p_j 
C_{\rmr}( |\psi_j\rangle \langle \psi_j|)
=
C_{\rmf}(\rho_A),\Label{E4-21B}
\end{align}
where $(a)$ follows from the fact that any decomposition of $\tilde{\rho}_A$ can be induced by a suitable POVM on $\cH_E$.

Let 
$\sM_{F_n}=
(U_{\rmi,n},\sM_{\rmc,n},F_n)$ be a sequence of incoherent strategies whose extraction rate is $R$ and which satisfies $\underline{d}_1(\sM_{F_n}|F_n)
\to 0$.
Since 
$\rE_{F_n} d_1(\sM_{F_n}(\sM_E (\tilde{\rho})^{\otimes n})) \to 0$
for any  local measurement $\sM_E$ of Eve,
using the same argument that leads to \eqref{5-2-1}, we can show the inequality 
\begin{align}
R =&
\liminf_{n\to \infty}
\rE_{F_n}\frac{1}{n} H(A|E)_{\sM_{F_n}(\sM_E(\tilde{\rho})^{\otimes n})} \nonumber \\
\le &
\max_{U_{\rmi}}H(A|E)_{U_{\rmi}( \sM_E(\tilde{\rho}) \otimes |0\rangle \langle 0|) U_{\rmi}^\dagger } .
\end{align}
Taking the minimum over $\sM_E$,
we have 
\begin{align}
R \le&
\min_{\sM_E} 
\max_{U_{\rmi}}H(A|E)_{U_{\rmi} (\sM_E(\tilde{\rho}) \otimes |0\rangle \langle 0|) U_{\rmi}^\dagger } \nonumber \\
=&\min_{\{p_j,|\psi_j\>\}} \sum_j p_j 
C_{\rmr}( |\psi_j\rangle \langle \psi_j|)
= C_{\rmf}(\rho_A),
\end{align}
which yields the opposite inequality to \eqref{E4-21B}.
\end{proofof}

\section{Possibility of extension of Theorem \ref{TH7}}
Since the coherence of formation is additive, that is, $ C_{\rmf}(\rho^{\otimes n})=n C_{\rmf}(\rho)$ \cite{WintY16},
one might expect a further extension of Theorem \ref{TH7}.
That is, one might speculate that the relation $ \overline{R}(\tilde{\rho}_A)=C_{\rmf}(\tilde{\rho}_A)$ 
holds even when the condition
$\underline{d}_1(\sM_{F_n}|F_n) \to 0$
is replaced by the stronger condition
$\max_{\sM_{E,n}} \rE_{F_n} d_1(\sM_{F_n}(\sM_{E,n} (\tilde{\rho}^{\otimes n}))) \to 0 $.
Here, note that 
$\sM_{E,n}$ is a POVM on the $n$-tensor product system; by contrast, in the definition of  $\underline{d}_1(\sM_{F_n}|F_n)$,
Eve's POVMs are restricted to tensor powers of POVMs on individual systems. However,
the additivity of $ C_{\rmf}$ alone  does not imply this stronger statement.

This stronger statement would follow from a stronger condition as described as follows. Given $\alpha >0$,
define the R\'{e}nyi coherence of formation as
\begin{align}
C_{\rmf,1/\alpha}(\rho):=
\min_{\{p_j,|\psi_j\rangle\}}
\frac{1}{1-\alpha}\log 
\sum_j p_j 
2^{(1-\alpha)C_{\rmr,1/\alpha}
( |\psi_j\rangle \langle \psi_j|)},
\end{align}
where $\{p_j,|\psi_j\rangle \}$ satisfies  $ \rho=\sum_j p_j |\psi_j\rangle \langle \psi_j|$. 
As shown later, 
\begin{align}
\lim_{\alpha \to 1}C_{\rmf,\alpha}(\rho)=C_{\rmf}(\rho). \Label{HT667}
\end{align}
In addition, if the classical R\'{e}nyi conditional entropy satisfied the chain rule, i.e.,
\begin{equation}\label{eq:ChainRuleRCE}
H_\alpha^\downarrow(A_1A_2|E)_\rho=H_\alpha^\downarrow(A_1|E)_\rho
+H_\alpha^\downarrow(A_2|A_1 E)_\rho,
\end{equation}
then the R\'{e}nyi coherence of formation would be additive,  
\begin{equation}\label{eq:AddRCF}
C_{\rmf,1/\alpha}(\rho_1 \otimes \rho_2)
= C_{\rmf,1/\alpha}(\rho_1)
+C_{\rmf,1/\alpha}(\rho_2)
\end{equation}
for any pair of  density matrices $\rho_1$  and $\rho_2$ on $A_1$ and $A_2$.
Assuming this additivity relation, we can show the inequality
\begin{align}
 \max_{\sM_{E,n}} \rE_{F_n} d_1\bigl(\sM_{F_n}^*(\sM_{E,n} (\tilde{\rho}^{\otimes n}))\bigr)
\le 
3 \cdot 2^{s n (R- C_{\rmf,{1-s}}(\tilde{\rho}_A )) } \Label{HT666}
\end{align}
for $s \in [0,1/2]$,
where $\sM_{F_n}^*$ denotes the optimal incoherent strategy composed of the computational-basis measurement and the application of the universal 2 hash function $F_n$. In this way,
the combination of  \eqref{HT667} and \eqref{HT666}  implies that
$\max_{\sM_{E,n}} \rE_{F_n} d_1(\sM_{F_n}^*(\sM_{E,n} (\tilde{\rho}^{\otimes n}))) \to 0 $ if $ R < C_{\rmf}(\tilde{\rho}_A)$.

However, it is known that the R\'{e}nyi conditional entropy does not satisfy
the chain rule satisfied by the usual conditional entropy even in the classical case \cite{Dupuis,FB}.
This quantity satisfies only a weaker version of the chain rule \cite[Theorem 1]{Dupuis}\cite[Corollary 87]{Ha4} \cite[Theorem 3]{FB}.
Hence, it is not easy to show 
the relation $ \overline{R}(\tilde{\rho}_A)=C_{\rmf}(\tilde{\rho}_A)$
with the above replacement.

\begin{proofof}{\eqref{HT667}}
For a given $\{p_j, |\psi_j\rangle\}$,
the value $\sum_j p_j 
2^{(1-\alpha)C_{\rmr,1/\alpha}( |\psi_j\rangle \langle |\psi_j|)}$
equals $1$ when $\alpha=1$.
So, the formula of the logarithmic derivative  
$\frac{d}{dx}\log f(x)= \frac{1}{\ln 2} \frac{d f}{dx}(x)/f(x)$
yields that
\begin{align*}
&\lim_{\alpha \to 1}
\frac{1}{1-\alpha}\log 
\sum_j p_j 
2^{(1-\alpha)C_{\rmr,1/\alpha}( |\psi_j\rangle \langle |\psi_j|)} \\
= &
\frac{1}{\ln 2}
\lim_{s \to 0}
\frac{
\sum_j p_j 
2^{s C_{\rmr,1/(1-s)}( |\psi_j\rangle \langle |\psi_j|)} 
-1}{s} \\
= &
\sum_j p_j 
\lim_{s \to 0}
\frac{
s C_{\rmr,1/(1-s)}( |\psi_j\rangle \langle |\psi_j|)}{s} \\
= &
\sum_j p_j 
C_{\rmr}( |\psi_j\rangle \langle |\psi_j|),
\end{align*}
which implies \eqref{HT667}.
\end{proofof}

\noindent{\it Derivation of \eqref{eq:AddRCF} assuming the chain rule \eqref{eq:ChainRuleRCE}:\;}When $\rho$ is pure,
according to \cite{Chitambar,ZHC}, we have
\begin{align}
C_{\rmf,1/\alpha}(\rho)=C_{\rmr,1/\alpha}(\rho)
= S_\alpha(\rho^{\diag})=S_\alpha(\sM_\rmc(\rho)),
\end{align}
where 
\begin{align}
S_{\alpha}(\rho)
= \frac{1}{1-\alpha} \log \tr(\rho^\alpha)
\end{align}
is the R\'enyi $\alpha$-entropy. 
So  the R\'{e}nyi coherence of formation can be expressed as 
\begin{align}
C_{\rmf,1/\alpha}(\rho)
=&
\min_{\{ p_j, \bar{\rho}_j \} }
\frac{1}{1-\alpha}\log 
\sum_j p_j 
2^{(1-\alpha)S_\alpha(\sM_\rmc(\bar{\rho}_j))} \nonumber \\
=&
\min_{\{ p_j, \bar{\rho}_j \} }
H_\alpha^\downarrow(A| J)_{\sM_\rmc (
\sum_{j} p_j \bar{\rho}_j \otimes |j \rangle \langle j| )} ,\Label{6-11}
\end{align}
where $\{p_j, \bar{\rho}_j \}$ satisfies  $ \rho=\sum_j p_j  \bar{\rho}_j$, and  $J$ denotes the  classical system of the register.
The expression \eqref{6-11} follows from the fact that the minimum is attained when all $\bar{\rho}_j$ are pure.

To prove \eqref{eq:AddRCF}, suppose $\rho_1\otimes \rho_2$ has an optimal pure-state decomposition  $\rho_1 \otimes\rho_2=\sum_j p_j \bar{\rho}_j$ such that
\begin{align}
C_{\rmf,1/\alpha}(\rho_1 \otimes \rho_2) 
=
\frac{1}{1-\alpha}\log 
\sum_j p_j 
2^{(1-\alpha)C_{\rmr,1/\alpha}
	( \bar{\rho}_j)}. 
\end{align}
Let $\sigma:= \sum_j  p_j \bar{\rho}_j\otimes |j\rangle \langle j|$; then 
\begin{align}
C_{\rmf,1/\alpha}(\rho_1 \otimes \rho_2) 
=
H_\alpha^\downarrow(A_1A_2|J)_{\sM_{\rmc,1}\otimes \sM_{\rmc,2}(\sigma)} ,
\end{align}
where 
$\sM_{\rmc,1}$ and  $\sM_{\rmc,2}$
express the computational-basis measurements on $A_1$ and  $A_2$, respectively.
Now, the  chain rule \eqref{eq:ChainRuleRCE} implies that 
\begin{align}
&C_{\rmf,1/\alpha}(\rho_1 \otimes \rho_2) 
=
H_\alpha^\downarrow(A_1A_2|J)_{\sM_{\rmc,1}\otimes \sM_{\rmc,2}(\sigma)} \nonumber\\
=&
H_\alpha^\downarrow(A_1|J)_{\sM_{\rmc,1}\otimes \sM_{\rmc,2}(\sigma)}
+
H_\alpha^\downarrow(A_2|A_1 J)_{\sM_{\rmc,1}\otimes \sM_{\rmc,2}(\sigma)} \nonumber\\
\ge & 
C_{\rmf,1/\alpha}(\rho_1 )+
C_{\rmf,1/\alpha}(\rho_2).
\end{align}
Since the opposite inequality
\begin{align}
&C_{\rmf,1/\alpha}(\rho_1 \otimes \rho_2) 
\le
C_{\rmf,1/\alpha}(\rho_1 )+
C_{\rmf,1/\alpha}(\rho_2)\Label{HTT}
\end{align}
is an easy consequence of the definition, we deduce \eqref{eq:AddRCF}, assuming that  the chain rule \eqref{eq:ChainRuleRCE} holds.

\noindent{\it Derivation of \eqref{HT666} assuming the additivity relation \eqref{eq:AddRCF}:\;}
When $\rho$ is a diagonal density matrix,
the paper \cite[Proposition 21]{epsilon} showed that
\begin{align}
\Delta_{d,2}(2^{R}|\rho)
\le
3 \cdot 2^{s R-s H_{\frac{1}{1-s}}^{\uparrow}(A|E)_\rho }
\le
3 \cdot 2^{s R-s H_{\frac{1}{1-s}}^{\downarrow}(A|E)_\rho }
\Label{3-26-1}
\end{align}
for $s \in [0,1/2]$.
Therefore,
\begin{align*}
& \max_{\sM_{E,n}} \rE_{F_n} d_1\bigl(\sM_{F_n}^*(\sM_{E,n} (\tilde{\rho}^{\otimes n}))\bigr)
\\
\stackrel{(a)}{\le}&
3\cdot 2^{s n R- s C_{\rmf,{1-s}}(\tilde{\rho}_A^{\otimes n} ) } 
\stackrel{(b)}{=}
3\cdot 2^{s n (R- C_{\rmf,{1-s}}(\tilde{\rho}_A )) }
\end{align*}
if $s \in [0,1/2]$ and $F_n$ is universal  2 hash. 
Here
$(a)$ follows from the combination of
\eqref{HE1}, \eqref{6-11}, and \eqref{3-26-1}, while $(b)$ follows from the additivity of the R\'{e}nyi coherence of formation.
\endproof

\section{Generalization of Theorem~\ref{TH1}}\Label{AK}
Before proving \thref{TH3X}, which characterizes the exponential decreasing rate of the leaked information, we need to 
generalize Theorem~\ref{TH1} in terms of R\'enyi conditional entropies and R\'enyi relative  entropies of coherence. 

The two types of R\'{e}nyi relative entropies defined in \eqref{eq:RREa} and \eqref{eq:RREb}  can be used to define two types of coherence measures \cite{Chitambar,ZHC},
\begin{align}
C_{\rmr,\alpha}(\rho):=\min_{\sigma\in \mathcal{I}}S_\alpha(\rho\|\sigma),\;\;
\underline{C}_{\rmr,\alpha}(\rho):=\min_{\sigma\in \mathcal{I}}\underline{S}_\alpha(\rho\|\sigma),\label{eq:H2616}
\end{align}
both of which increase monotonically with $\alpha$. The following theorem generalizes Theorem~\ref{TH1} and thereby demonstrates the significance of these  R\'enyi relative  entropies of coherence.
\begin{theorem}\Label{TH2}
	\begin{align}
&	\frac{1}{n}\max_{\Lambda_\rmi} 
	\overline{H}_\alpha^{\uparrow}(A|E)_{\tilde{\rho}^{\otimes n}[\Lambda_\rmi]}
	= 
	\max_{\Lambda_\rmi} 
	\overline{H}_\alpha^{\uparrow}(A|E)_{\tilde{\rho}[\Lambda_\rmi]}\nonumber \\
	= &
	\max_{U_\rmi} 
	\overline{H}_\alpha^{\uparrow}(A|E)_{\tilde{\rho}[U_\rmi]}
    =
	\overline{H}_\alpha^{\uparrow}(A|E)_{\tilde{\rho}[U_{\mrm{CNOT}}]}
	= 
	\underline{C}_{\rmr,\beta}(\tilde{\rho}_A),
	\Label{eq:H2610} \\
&	\frac{1}{n}\max_{\Lambda_\rmi} 
	\overline{H}_\alpha^{\downarrow}(A|E)_{\tilde{\rho}^{\otimes n}[\Lambda_\rmi]}
	= 
	\max_{\Lambda_\rmi} 
	\overline{H}_\alpha^{\downarrow}(A|E)_{\tilde{\rho}[\Lambda_\rmi]}\nonumber \\
	= &
	\max_{U_\rmi} 
	\overline{H}_\alpha^{\downarrow}(A|E)_{\tilde{\rho}[U_\rmi]}
	=
	\overline{H}_\alpha^{\downarrow}(A|E)_{\tilde{\rho}[U_{\mrm{CNOT}}]}
	= C_{\rmr,\beta}(\tilde{\rho}_A),
	\Label{eq:H2611}
	\end{align}
	where $U_\rmi$ is an incoherent unitary,  $\Lambda_\rmi$ is an incoherence-preserving operation,
	\eqref{eq:H2610} holds for  $\alpha,\beta \in[\frac{1}{2},\infty]$ with $ \frac{1}{\alpha}+\frac{1}{\beta}=2$, 
	while
	\eqref{eq:H2611} holds for
	$\alpha \in [\frac{1}{2},\infty]$ and
	$\beta \in [0,2]$ with  $ \alpha\beta=1$.
\end{theorem}

The proof of Theorem~\ref{TH2} relies on the duality relations between R\'{e}nyi conditional entropies.
When $\rho$ is a pure state across the three systems $\cH_A,\cH_B$, and $\cH_E$,
these conditional entropies obey the following duality relations 
\cite{MullDSF14}\cite{Ha3}\cite[Theorem 5.13]{Haya17book}\cite{C6}:
\begin{align}
\overline{H}_\alpha^{\uparrow}(A|E)_{\rho}+
\overline{H}_\beta^{\uparrow}(A|B)_{\rho} &=0 \Label{eq:H267},\\
\overline{H}_\alpha^{\downarrow}(A|E)_{\rho}
+{H}_\beta^{\uparrow}(A|B)_{\rho}&=0,
\Label{eq:H268}
\end{align}
where 
\eqref{eq:H267} holds for $\alpha,\beta \in[\frac{1}{2},\infty]$  with $ \frac{1}{\alpha}+\frac{1}{\beta}=2$
and  
\eqref{eq:H268} holds for  $\alpha,\beta \in [0,\infty]$  with $ \alpha\beta=1$.

\begin{proofof}{Theorem \ref{TH2}}
Let $\beta=\alpha/(2\alpha-1)$, then $ \frac{1}{\alpha}+\frac{1}{\beta}=2$. Let  $U_\rmi$ be any incoherent unitary acting on $\caH_A\otimes \caH_B$; then 	
	\begin{align}
&\overline{H}_\alpha^{\uparrow}(A|E)_{\tilde{\rho}[U_\rmi]}
 =-\overline{H}_\beta^{\uparrow}(A|B)_{\tilde{\rho}[U_\rmi]}\le  \underline{E}_{\rmr,\beta}(\tilde{\rho}[U_\rmi]_{AB})\nonumber \\
 &\leq \underline{C}_{\rmr,\beta}(\tilde{\rho}[U_\rmi]_{AB})
	\leq \underline{C}_{\rmr,\beta}(\tilde{\rho}_A).
	\Label{eq:H271T}
	\end{align}
Here the equality follows from \eqref{eq:H267}, the	first inequality follows from \cite[Lemma~4]{ZHC}, and the other two inequalities are trivial. According to \cite[Theorem 1]{ZHC}, the upper bound in the RHS of \eqref{eq:H271T} is attained when 
	$U_\rmi$ is the generalized CNOT gate, in which case
	$\tilde{\rho}[U_\rmi]_{AB}$
	is maximally correlated.

By the same reasoning as above, we deduce the equality
	$\max_{\Lambda_\rmi} 
	\overline{H}_\alpha^{\uparrow}(A|E)_{\tilde{\rho}[\Lambda_\rmi]} =
	\underline{C}_{\rmr,\beta}(\tilde{\rho}_A)$, which in turn implies that $\max_{\Lambda_\rmi} 
	\overline{H}_\alpha^{\uparrow}(A|E)_{\tilde{\rho}^{\otimes n}[\Lambda_\rmi]}
	= 
	\underline{C}_{\rmr,\beta}(\tilde{\rho}_A^{\otimes n})$. Now the proof of \eqref{eq:H2610} is completed by the additivity relation 
$	\underline{C}_{\rmr,\beta}(\tilde{\rho}_A^{\otimes n})=	n \underline{C}_{\rmr,\beta}(\tilde{\rho}_A)$, which  is shown in \cite[Theorem~3]{ZHC}.

Finally, \eqref{eq:H2611} can be proved in a similar  way.
\end{proofof}

\section{Proof of Theorem \ref{TH3X}}\Label{ATH3}

\begin{proof}[Proof of \thref{TH3X}]
Applying Proposition \ref{LH1} in Appendix~\ref{A2A}, we deduce that
	\begin{align}
	&\liminf_{n \to \infty}\frac{-1}{n}\log 
	\rE_{F_n} d_1\bigl(\sM_{F_n}^*(\tilde{\rho}^{\otimes n} )\bigr) \nonumber \\
	\ge & \max_{s \in [0,1]}\frac{1}{2} \bigl(s \overline{H}_{1+s}^\uparrow
	(A|E)_{\sM_\rmc(\tilde{\rho})}-sR\bigr).
	\Label{Eq4-20}
	\end{align}
Now  Theorem~\ref{TH3X} is a corollary of  the following equation 
	\begin{align}\overline{H}_{1+s}^\uparrow
	(A|E)_{\sM_\rmc(\tilde{\rho})}=\overline{H}_{1+s}^\uparrow
	(A|E)_{\tilde{\rho}[U_\mrm{CNOT}]}=
	\underline{C}_{\rmr,\frac{1+s}{1+2s}}(\tilde{\rho}_A),
	\end{align}
where the second equality follows from  \eqref{eq:H2610} in \thref{TH2}.
\end{proof}

\section{Security analysis based on an alternative criterion}
\Label{AAT}
Here we analyze the  exponential decreasing rate of the alternative   security measure $I'({\rho}_{AE})$ defined in \eqref{eq:AltSC}, which denotes
the relative entropy between the true state and the ideal state \cite{CN}. 
Similar to Theorem~\ref{TH3X}, we have
\begin{align}
&\liminf_{n \to \infty}\frac{-1}{n}\log 
\rE_{F_n}I'\bigl(\sM_{F_n}^*(\tilde{\rho}^{\otimes n} )\bigr) \nonumber \\
\ge & \max_{s \in [0,1]}\bigl(s 
C_{\rmr,\frac{1}{1+s}}(\tilde{\rho}_A)
-sR\bigr).\Label{E-29-12}
\end{align}
Again, the exponential decreasing rate of the leaked information is controlled by 
R\'{e}nyi relative entropies of coherence.
To prove \eqref{E-29-12}, note that
\begin{align}
&\liminf_{n \to \infty}\frac{-1}{n}\log 
\rE_{F_n} I'\bigl(\sM_{F_n}^*(\tilde{\rho}^{\otimes n} )\bigr) \nonumber \\
\ge & \max_{s \in [0,1]}\bigl(s \overline{H}_{1+s}^\downarrow
(A|E)_{\sM_\rmc(\tilde{\rho})}-sR\bigr)
\end{align}
according to Proposition \ref{LH2} in Appendix \ref{A2A}. Now  \eqref{E-29-12} is a corollary of  the following equation 
\begin{align}\overline{H}_{1+s}^\downarrow
(A|E)_{\sM_\rmc(\tilde{\rho})}=\overline{H}_{1+s}^\downarrow
(A|E)_{\tilde{\rho}[U_\mrm{CNOT}]}=
C_{\rmr,\frac{1}{1+s}}(\tilde{\rho}_A),
\end{align}
where the second equality follows from  \eqref{eq:H2611} in \thref{TH2}.

\section{Proof of Corollary~\ref{cor:HZL}}\Label{AHZL}
In view of Theorems \ref{TH3} and \ref{TH7}, Corollary~\ref{cor:HZL} is an immediate consequence of the following lemma.
\begin{lemma}\Label{cor:HZL2}
	A qubit state $\rho$ saturates the inequality
	$C_{\rmf}(\rho)\geq C_{\rmr}(\rho)$ iff $\rho$ is pure or incoherent.
\end{lemma}

\begin{proof}
	The inequality $C_{\rmf}\geq C_{\rmr}$ holds in general because $C_{\rmf}$ is the convex roof of $C_{\rmr}$. 
	
Any qubit state can be written as follows,
	\begin{equation}\Label{eq:rhoQubit}
	\rho=\frac{1}{2}(I+x\sigma_x+y\sigma_y+z\sigma_z),\quad x^2+y^2+z^2\leq 1.
	\end{equation}
	Let $r=\sqrt{x^2+y^2+z^2}$, then 
	\begin{align}
	C_{\rmr}(\rho)&=H\Bigl(\frac{1+z}{2}\Bigr)-H\Bigl(\frac{1+r}{2}\Bigr),\\
	C_{\rmf}(\rho)&=H\Bigl(\frac{1+\sqrt{1-x^2-y^2}}{2}\Bigr),
	\end{align}
	where $H(p)=-p\log p-(1-p)\log(1-p)$, and the formula for $C_{\rmf}(\rho)$ was derived in \rcite{YZCM}. 
The relation between $C_{\rmr}$ and $C_{\rmf}$ was  illustrated in Fig. 3 of \cite{YZCM2}.

If $\rho$ is pure, then $x^2+y^2+z^2=1$,  so that $C_{\rmf}(\rho)=C_{\rmr}(\rho)=H\bigl(\frac{1+z}{2}\bigr)$. If $\rho$ is incoherent, then $x=y=0$, so that $C_{\rmf}=C_{\rmr}=0$.

To determine the condition for saturating the inequality  $C_{\rmf}\geq C_{\rmr}$, first consider the case $y=z=0$, so that $C_{\rmr}(\rho)=1-H\bigl(\frac{1+x}{2}\bigr)$ and $C_{\rmf}(\rho)=H\bigl(\frac{1+\sqrt{1-x^2}}{2}\bigr)$.
	By computing the first and second derivatives of $C_{\rmf}(\rho)-C_{\rmr}(\rho)$ with $x$, it is not difficult to prove that $C_{\rmf}(\rho)=C_{\rmr}(\rho)$ iff $x=0$ or $x=\pm1$.

Next, consider the case $0<x<1$, $y=0$, $z\geq 0$, and $x^2+z^2<1$. Let $\rho_1, \rho_2$ be two qubit states with Bloch vectors $(x, 0,0)$ and $(x,0,\sqrt{1-x^2})$, respectively. Then $\rho$ is a convex combination of $\rho_1$ and $\rho_2$, that is, $\rho=p_1\rho_1+p_2\rho_2$ with $p_1>0$. 
	In addition,
	\begin{align}
	C_{\rmf}(\rho)=C_{\rmf}(\rho_1)=C_{\rmf}(\rho_2)=C_{\rmr}(\rho_2),\quad
	C_{\rmf}(\rho_1)>C_{\rmr}(\rho_1). 
	\end{align} 
	Given that $C_{\rmr}$ is convex, we conclude that 
	\begin{align}
	C_{\rmr}(\rho)&\leq p_1C_{\rmr}(\rho_1)+p_2 C_{\rmr}(\rho_2)< p_1C_{\rmf}(\rho_1)+p_2C_{\rmf}(\rho_2)\nonumber\\
	&=C_{\rmf}(\rho). 
	\end{align}
	
	By symmetry $C_{\rmr}(\rho)<C_{\rmf}(\rho)$ whenever $x^2+y^2+z^2<1$ and $x^2+y^2>0$. Therefore, the inequality $C_{\rmf}(\rho)\geq C_{\rmr}(\rho)$ is saturated iff the qubit state $\rho$ is  pure or incoherent.
\end{proof}

\section{Application to  quantum random number generators}\Label{industrial}

In this appendix, we explain the application of our study to the design of a quantum random number generator.
Remember that our optimal incoherent strategy can be realized by the measurement $\sM_{\rmc}$ in
the computational basis  followed by classical data processing.
A quantum random number generator consists of the following ingredients:  an internal quantum system, the device that performs the computational-basis measurement,
and the data processor that  extracts secure uniform random numbers. To  make a quantum random number generator as an industrial product, the supplier
needs to specify the method for  preparing the state of the internal system, which can  be identified by quantum state tomography \cite{Hol,Gill,HayaB}\cite[Chapter 6]{Haya17book}.

To implement quantum state tomography, the supplier can  apply suitable  measurements to the quantum system and reconstruct the quantum state based on the measurement statistics. Since quantum measurements are  destructive, to achieve sufficient precision in this procedure,  usually many identically-prepared quantum states are needed to gather enough information. In addition,  quantum state  tomography may require operations that are not incoherent, but this is not a problem. Note that in the design stage of the random number generator, it is reasonable to assume that the supplier can access certain advanced equipments and are not restricted to  incoherent operations, in contrast with the user stage of the device. 


Once the internal state $\tilde{\rho}_A$ of the random number generator is determined,
the supplier can choose the parameter $n$ and the extraction rate $R$ 
based on the upper bound determined by \eqref{HE1} and \eqref{HE2}, so that the amount of leaked information $d_1(\sM_{F_n}^* |F_n)$ is less than a given threshold.  
Note that the combination of \eqref{HE1} and \eqref{HE2} allows to perform a finite-length analysis.
Since $d_1(\sM_{F_n}^* |F_n)$ decreases exponentially with $n$, this task can be achieved with a suitable choice of the parameters as long as  $\tilde{\rho}_A$ is sufficiently coherent. 
In addition, the supplier needs to design universal 2 hash functions so as to perform randomness extraction. Although random numbers are needed to apply random hash functions, the security is not compromised even if Eve knows which specific hash function is applied each time. 
Therefore, the supplier needs to design universal 2  hash functions only once, which can then be reused repeatedly. 
The user does not have to invest random numbers to operate the random number generator.

So far we have assumed that the measurement device for estimating the quantum state of the internal system is trustworthy. 
This assumption is not absolutely necessary. 
Even when the measurement device cannot be trusted, 
the supplier can identify the measurement device by applying the method of self testing
\cite{MY2,McKague1,HH}.
The self testing was originally proposed using the CHSH test, which requires the preparation of a Bell state \cite{MY2,McKague1}.
Recently, the paper \cite{HH} improved it by proposing a hybrid method of the CHSH test and the Bell state test, i.e., the stabilizer test. Applying this method before quantum state tomography, 
the supplier can identify the measurement device so that the quantum state of the internal system can be guaranteed. 
Similarly, the measurement device for generating  random numbers may not be trustworthy.
In that case, the supplier can apply the self testing to this measurement device.
In this way, the supplier can guarantee the security of the random numbers generated
by the random number generator. 

\bigskip

\section{Relation with \cite{YZCM}}\Label{Relation}
Here, we need to discuss the relation with the paper~\cite{YZCM}, which studied a related but different problem. The focus of the current paper is the extraction of uniform random numbers by incoherent strategies, which include the measurement on the computational basis and general incoherent operations (or incoherence-preserving operations). The focus of \cite{YZCM} is the connection between intrinsic randomness and coherence measures.
In information theory, the term ``intrinsic randomness'' usually  means the extraction of uniform random numbers \cite{Ha5,Han}. In \cite{YZCM},  the term has a related but different meaning, that is, the randomness of measurement outcomes conditioned on Eve's prediction. With this latter interpretation,
\cite{YZCM} showed that the intrinsic randomness of measurement outcomes with respect to the  computational basis is equal  to the coherence of formation, without discussing general protocols for extracting uniform random numbers.

\end{document}